\documentclass[10pt]{article}
\usepackage[affil-it]{authblk}

\usepackage[utf8]{inputenc}
\usepackage[T1]{fontenc}
\usepackage{adjustbox}
\usepackage{algpseudocode}
\usepackage{algorithm}
\usepackage{amsmath}
\usepackage{amsfonts}
\usepackage{amssymb}
\usepackage{amsthm}
\usepackage[english]{babel}
\usepackage{blindtext}
\usepackage{booktabs}
\usepackage{cancel}
\usepackage{chngcntr}
\usepackage{color}
\usepackage{colortbl}
\usepackage{diagbox}
\usepackage{enumerate}
\usepackage{epsfig}
\usepackage{fancyhdr}
\usepackage{float}
\usepackage{float}
\usepackage{framed}
\usepackage[bottom]{footmisc}
\usepackage{graphicx}
\usepackage{hyperref}
\usepackage{import}
\usepackage{latexsym}
\usepackage{listings}
\usepackage{lscape}
\usepackage{makeidx}
\usepackage{marginnote}
\usepackage{mathtools}
\usepackage{microtype}
\usepackage{multicol}
\usepackage{multirow}
\usepackage{nicefrac}
\usepackage{paralist}
\usepackage{pslatex}
\usepackage{relsize}
\usepackage{rotating}
\usepackage{setspace}
\usepackage{soul}
\usepackage{subfigure}
\usepackage{tabularx}
\usepackage{threeparttable}
\usepackage{titlesec}
\usepackage{todonotes}
\usepackage{upgreek}
\usepackage{wasysym}
\usepackage{wrapfig}
\usepackage{xcolor}
\usepackage{xfrac}
\usepackage{tikz,tikz-qtree}
\usetikzlibrary{positioning,arrows,arrows,automata,shapes}
\tikzset{
  invisible/.style={opacity=0},
  visible on/.style={alt={#1{}{invisible}}},
  alt/.code args={<#1>#2#3}{%
    \alt<#1>{\pgfkeysalso{#2}}{\pgfkeysalso{#3}}
  },
}

\newcommand{\pend}{P}				% undirected pendant
\newcommand{\utt}{TT}				% undirected true twin
\newcommand{\uft}{FT}				% undirected false twin
\newcommand{\ft}{\diamond}			% directed false twin
\newcommand{\tdt}{\leftrightarrow}	% true double twin
\newcommand{\tit}{\to}				% true incoming twin
\newcommand{\tot}{\leftarrow}			% true outgoint twin
\newcommand{\ppv}{+}					% Pendant plus vertex
\newcommand{\pmv}{-}					% Pendant minus vertex

\newcommand{\un} {{\it un}}
\newcommand{\UDH}{\text{DH}}
\newcommand{\DDH} {\text{DDH}}

\newcommand{\dpws} {\text{d-pw}}

\newcommand{\dtws} {\text{d-tw}}

\newcommand{\dagws} {\text{dagw}}

\newcommand{\cyr} {\text{cr}}
\newcommand{\mso} {\text{MSO}}
\theoremstyle{definition}
\newtheorem{definition}{Definition}[section]
\theoremstyle{plain}
\newtheorem{theorem}[definition]{Theorem}
\newtheorem{lemma}[definition]{Lemma}

\newtheorem{corollary}[definition]{Corollary}
\newtheorem{proposition}[definition]{Proposition}
\theoremstyle{remark}
\newtheorem{example}[definition]{Example}

\begin{document}
%%%%%%%%%%%%%%%%%%%%%%%%%%%%%%%%%%%%%%%%%%%%%%%%%%%%%%%%%%%%%%%%%%%%%%%%%%%
%%%%%%%%%%%%%%%%%%%%%%%%%%%%%%%%%%%%%%%%%%%%%%%%%%%%%%%%%%%%%%%%%%%%%%%%%%%

\title{Twin-Distance-Hereditary Digraphs}

\author{Dominique Komander}
\author{Carolin Rehs}

\affil{Heinrich-Heine-University Duesseldorf,
Institute of Computer Science,
Algorithmics for Hard Problems Group,
40225 Duesseldorf, Germany}

\maketitle

\pagestyle{plain}

\begin{abstract}
We investigate structural and algorithmic advantages of a directed version of the well-researched class of distance-hereditary graphs.
Since the previously defined distance-hereditary digraphs do not permit a recursive structure, we define directed twin-distance-hereditary graphs, which can be constructed by several twin and pendant vertex operations analogously to undirected distance-hereditary graphs and which still preserves the distance hereditary property.
We give a characterization by forbidden induced subdigraphs and place the class in the hierarchy, comparing it to related classes.
We further show algorithmic advantages concerning directed width parameters, directed graph coloring and some other well-known digraph problems which are NP-hard in general, but computable in polynomial or even linear time on twin-distance-hereditary digraphs.
This includes computability of directed path-width and tree-width in linear time and the directed chromatic number in polynomial time.
From our result that directed twin-distance-hereditary graphs have directed clique-width at most $3$ it follows by Courcelle's theorem on directed clique-width that we can compute every graph problem describable in which is describable in monadic second-order logic on quantification over vertices and vertex sets as well as some further problems like Hamiltonian Path/Cycle in polynomial time.
\end{abstract}

%%%%%%%%%%%%%%%%%%%%%%%%%%%%%%%%%%%%%%%%%%%%%%%%%%%%%%%%%%%%%%%%%%%%%%%%%%%
%%%%%%%%%%%%%%%%%%%%%%%%%%%%%%%%%%%%%%%%%%%%%%%%%%%%%%%%%%%%%%%%%%%%%%%%%%%
%%%%%%%%%%%%%%%%%%%%%%%%%%%%%%%%%%%%%%%%%%%%%%%%%%%%%%%%%%%%%%%%%%%%%%%%%%%

%%%%%%%%%%%%%%%%%%%%%%%%%%%%%%%%%%%%%%%%%%%%%%%%%%%%%%%%%%%%%%%%%%%%%%%%%%%%%%%%%%%%%%%%%%%%%%%%%%%
\section{Introduction}
%%%%%%%%%%%%%%%%%%%%%%%%%%%%%%%%%%%%%%%%%%%%%%%%%%%%%%%%%%%%%%%%%%%%%%%%%%%%%%%%%%%%%%%%%%%%%%%%%%%
Distance-hereditary graphs have been introduced by Howorka in 1977 \cite{How77}. They are exactly the graphs which are distance-hereditary for their connected induced subgraphs, which means that if any two vertices $u$ and $v$ belong to a connected induced subgraph $H$ of a graph $G$, then some shortest path between $u$ and $v$ in $G$ has to be a subgraph of $H$.
But this is not the only definition of distance-hereditary graphs. Most important from an algorithmic perspective are the definition by forbidden induced subgraphs and the recursive construction by twins and pendant vertices. That is, a distance-hereditary graph can be defined recursively from a single vertex by the following three operations:
\begin{enumerate}
\item Adding a pendant vertex, which is a vertex with only one edge to an existent vertex,
\item  adding a false twin, which is a vertex with the same neighborhood as an existent vertex and no edge to this vertex and
\item adding a true twin, which is a vertex with the same neighborhood as an existent vertex and an edge to this vertex.
\end{enumerate}
Attempting to define a directed version of distance-hereditary graphs, it is necessary to decide which of these definitions modified to a directed definition is most promising to give an useful digraph class. We consider this matter concerning some directed graph parameters.
In \cite{LR10}, the authors use a straightforward way in generalizing the property of distance-hereditary, i.e. the property that some shortest path has to be induced subgraph, on digraphs. But their graphs are limited to oriented digraphs,which are digraphs without bidirectional edges. For undirected distance-hereditary graphs, tree-width is computable in linear time \cite{BDK00}. Further, linear rank-width of distance-hereditary graphs is computable in polynomial time \cite{AKK17}. The clique-width of any distance-hereditary graph is at most $3$ \cite{GR99}, but path-width is hard even on bipartite distance-hereditary graphs \cite{KBMK93}.

In this paper we introduce a directed version of distance-hereditary graphs, which differs from the already known distance-hereditary digraphs from \cite{LR10}. We preserve the distance-hereditary property for our new class of directed twin-distance-hereditary graphs (twin-dh digraphs for short) but we expand it as we allow bidirectional edges. These twin-dh digraphs are generated by a directed pruning sequence, as in the undirected class, by using twins and pendant vertices. This structure is algoritmically useful for showing that twin-dh digraphs have bounded clique-width. After the definition by a directed pruning sequence we go on with an other characterization of this class. Distance-hereditary graphs can be characterized by forbidden induced subgraphs, so we provide as well such a characterization for twin-dh digraphs.

We show how to place the class in the hierarchy of related common directed graph classes and conclude that the class of twin-dh digraphs is a subclass of extended directed co-graphs, which allows to deduce some properties. Further, we take a closer look to show the connection between directed co-graphs and twin-dh digraphs and compare the class to the previously defined distance-hereditary digraphs from \cite{LR10}.

Moreover, we investigate directed width parameters on this graph class.
By showing that every strong component of a twin-dh digraphs is a directed co-graph, we can prove that there are linear time algorithms to compute directed path-width, directed tree-width, DAG-width, and cycle rank on twin-dh digraphs. This further reproves the equality of these parameters, which is already given by the fact that twin-dh digraphs are also extended directed co-graphs. It generalizes our results on directed co-graphs in \cite{GKR21a}. Furthermore, we present some properties which demonstrate the usability of the class. Showing that twin-dh digraphs have directed clique-width at most $3$, it follows that for every digraph problem expressible in monadic second order logic with quantification over vertices and vertex sets there exists an fpt-algorithm with respect to the parameter directed clique-width. Thus, we can get polynomial time solutions for several different problems. From the bounded directed clique-width we can also follow that we can solve problems like Directed Hamiltonian Path, Directed Hamiltonian Cycle, Directed Cut, and Regular Subdigraph in polynomial time, following from \cite{GWY16}. From our results in \cite{GKR21} we conclude that we also can solve the dichromatic number problem in polynomial time on this class.

%%%%%%%%%%%%%%%%%%%%%%%%%%%%%%%%%%%%%%%%%%%%%%%%%%%%%%%%%%%%%%%%%%%%%%%%%%%%%%%%%%%%%%%%%%%%%%%%%%%
\section{Preliminaries}
%%%%%%%%%%%%%%%%%%%%%%%%%%%%%%%%%%%%%%%%%%%%%%%%%%%%%%%%%%%%%%%%%%%%%%%%%%%%%%%%%%%%%%%%%%%%%%%%%%%
We use the notations of Bang-Jensen and Gutin \cite{BG09} for graphs and digraphs.
By the term digraph, we always mean directed graphs with neither multi-edges nor loops.
For a (directed) graph $G$, $V(G)$ is the vertex set and $E(G)$ the edge (arc) set of $G$.
By $\un(G)$ we define the underlying undirected graph of a digraph, which can be obtained by replacing every arc $(u,v)$ by $\{u,v\}$ and deleting multi-edges. For a subset $V'\subseteq V(G)$ of a digraph $G$ let $G[V']$ be the subgraph induced by the vertex set $V'$.
The \emph{neighborhood} of a vertex $v$ in an undirected graph $G$ is $N_G(v)=\{u \mid \{u,v\} \in E(G)\}$. If the graph is clear from the context, we simply say $N(v)$. For the directed neighborhood we distinguish between the incoming-neighbors $N^-(v)=\{u \mid (u,v)\in E(G)\}$ and the outgoing-neighbors $N^+(v)=\{u \mid (v,u)\in E(G)\}$.
For a digraph $G$, a \emph{strongly connected subdigraph} is an induced subdigraph $H$ of $G$ such that for all vertices $u,v \in V(H)$ there is a directed walk from $u$ to $v$ in $H$ and a directed walk from $v$ to $u$ in $H$.
A \emph{strong component} of $G$ is a maximal strongly connected subdigraph of $G$, i.e. a strong component $H$ of $G$ such that there is no vertex $v \in V(G) \setminus V(H)$ such that the induced subdigraph of $G$ generated by $V(H) \cup \{v\}$ is a strong component of $G$.
A digraph $G$ is \emph{weakly connected} if $\un(G)$ is connected.
A \emph{biorientation} of an undirected graph $G$ is obtained by replacing every edge $\{u,v\}\in E(G)$ by one or both of the arcs $(u,v)$ and $(v,u)$.
We call vertex $v$ in digraph $G$ a \emph{bioriented leaf} if there exist $(u,v),(v,u)\in E(G)$ and if $v$ is a leaf in $\un(G)$, which means that in $un(G)$ it holds that $\mid N_{un(G)}(v)\mid = 1$.

A \emph{ (directed) graph parameter} of a (directed) graph $G$ is a function $\alpha$ which maps from graph $G$ to an integer. Two graph parameters $\alpha$ and $\beta$ are called \emph{equivalent}, if there are some functions $f,g$ such that for every digraph $G$ it holds $\alpha(G)\leq f(\beta(G))$ and $\beta(G) \leq g(\alpha(G))$.

\subsection{Distance-Hereditary Graphs and Co-Graphs}

We now recall several different definitions of undirected distance-hereditary graphs \cite{BM86, Oum05a}.
We therefore need to define the terms of pendant vertices and twins. Let $G$ be an undirected, connected graph. A vertex $v \in V(G)$ is called \emph{pendant} if there is $u \in V(G)$ such that $\{u,v\} \in E(G)$ and for all other $w \in V(G)$, $w \neq u$ it holds that $\{w,v\} \not\in E(G)$.
A vertex $v \in V(G)$ is called \emph{twin} of $u \in V(G)$, if $N(u) \setminus \{v\} = N(v) \setminus \{u\}$.
It is called \emph{true twin} if $\{u,v\} \in E(G)$, otherwise \emph{false twin}.

The following conditions are equivalent:
\begin{enumerate}
\item $G$ is distance-hereditary, $G\in DH$ for short, that is, for every two vertices $u$ and $v$, all induced $u,v$-paths have the same length.
\item For every two vertices $u$ and $v$ that have distance $2$ to each other, there is no induced path between $u$ and $v$ of length greater than $2$.
\item{\label{lem:hhdg-free}} The house, holes, domino, and gem (see Fig. \ref{fig:forbidden_DH}) are not induced subgraphs of $G$.

\item $G$ can be defined recursively from a single vertex by adding twins and pendant vertices.
\end{enumerate}

We extend this recursive definition to digraphs and therefore, we recall the so-called pruning sequence to define a distance-hereditary graph $G$. Let $\sigma = (v_0, \dots v_{n-1})$ be an ordering on $V(G)$. Then $S=(s_1, \dots, s_{n-1})$ is the pruning sequence of $G$, where for every $1 \leq j \leq i \leq n-1$, $s_i$ is one of the following:
 $(x_i, \pend, x_j)$ for a $x_i$ a pendant vertex of $x_j$,  $(x_i, \utt, x_j)$ for a $x_i$ a true twin of $x_j$, $(x_i, \uft, x_j)$ for a $x_i$ a false twin of $x_j$.

Co-graphs have been introduced in the 1970s by a number of authors under different notations, e.g. in \cite{Ler71}.
A co-graph $G$ can be obtained starting with single vertices by applying the disjoint union operation, where we just add two vertex-disjoint graphs and the join operation, where we add all possible edges between two disjoint graphs.
An equivalent definition is that $G$ can be obtained from single vertices by taking disjoint union and adding twins.
Thus, it becomes obvious that co-graphs are a subclass of distance-hereditary graphs.

\subsection{Directed Co-Graphs}

Directed Co-Graphs has been considered a lot in literature. We here recall the most common definition:

\begin{definition}[Directed co-graphs, \cite{CP06}]
The class of {\em directed co-graphs} is recursively defined as follows.
\begin{enumerate}[(i)]
\item Every digraph on a single vertex $(\{v\},\emptyset)$,
denoted by $\bullet$, is a {\em directed co-graph}.

\item If $G_1$ and $G_2$ are vertex-disjoint directed co-graphs, then
\begin{enumerate}
\item
the disjoint union
$G_1\oplus G_2$, i.e. the digraph with vertex set $V(G_1)\cup  V(G_2)$ and
arc set $E(G_1)\cup  E(G_2)$,

\item
the series composition
$G_1 \otimes  G_2$, i.e. defined by their disjoint union plus all possible arcs between
vertices of $G_1$ and $G_2$, and
\item
the order composition
$G_1\oslash  G_2$, i.e. defined by their disjoint union plus all possible arcs from
vertices of $G_1$ to vertices of $G_2$, are {\em directed co-graphs}.
\end{enumerate}
\end{enumerate}
\end{definition}

As in \cite{GKR21a} we also consider the class of \emph{extended directed co-graphs}, that includes additionally to the operations above the so-called directed union transformation, which is the disjoint union of $G_1$ and $G_2$ plus possible additional arcs from $G_1$ to $G_2$. An expression for an (extended) directed co-graph is a series of these procedures. Nevertheless, these series are not unique representations of the resulting digraphs, since we have no information about which edges are inserted by the directed union.

For every directed co-graph $G$ the underlying undirected graph $\un(G)$ is a co-graph, but not every orientation of an undirected co-graph is a directed co-graph.
In \cite{CP06} it has been shown
that directed co-graphs can be characterized by eight forbidden induced
subdigraphs.

%%%%%%%%%%%%%%%%%%%%%%%%%%%%%%%%%%%%%%%%%%%%%%%%%%%%%%%%%%%%%%%%%%%%%%%%%%%%%%%%%%%%%%%%%%%%%%%%%%%
\section{Directed Distance-Hereditary Graphs}
%%%%%%%%%%%%%%%%%%%%%%%%%%%%%%%%%%%%%%%%%%%%%%%%%%%%%%%%%%%%%%%%%%%%%%%%%%%%%%%%%%%%%%%%%%%%%%%%%%%

We now come to define a directed version distance-hereditary graphs. A straight-forward idea given by the name of the graph class is, to say that a digraph $G$ is called distance-hereditary, if for every induced subdigraph $H$ of $G$ and for every vertices $u,v$ in $H$, the shortest path between $u$ and $v$ in $H$ has the same length as the shortest path between $u$ and $v$ in $G$. This idea has been pursued in \cite{LR10} but only for oriented digraphs without bioriented edges \cite{Schrader}.

In the following, we generalize the recursive definition by twins and pendant vertices to digraphs, which admits several algorithmic results.

There are at least three different definitions of twins in digraphs. In~\cite{KR09}, twins have been defined to obtain distance-hereditary digraphs in context of directed rank-width and split decomposition. Thus, \cite{KR09} can be seen as an attempt to extend undirected distance-hereditary graphs to directed distance-hereditary graphs.
In~\cite{FoucaudHP19}, twins have been defined to obtain results about domination and location-domination, and in~\cite{GutinY02} (see also~\cite[p. 282]{BG18}) in studying diameter in digraphs. In \cite{LR10} twins are introduced in context
of a distance based directed version of distance-hereditary graphs, but they do not lead to a characterization of this graph class.

We define directed twins and pendant vertices in digraphs as follows.

 \begin{definition}\label{def:twins}
 Let $G$ be a directed graph.
 \begin{itemize}
 \item Vertices $x,y \in V(G)$ are \emph{directed twins}\footnote{We say twins for short, but the meaning is directed twins if the context is a digraph.} if $N^{-}(x)\setminus\{y\}=N^{-}(y)\setminus\{x\}$ and $N^{+}(x)\setminus\{y\}=N^{+}(y)\setminus\{x\}$. We distinguish between
 \begin{itemize}
 	\item $x$ is a (directed) \emph{false twin} ($\ft$) of $y$, if $(x,y), (y,x) \not\in E(G)$.
  \item $x$ is a \emph{true out-twin} ($\tot$) of $y$ if $(y,x) \in E(G)$, $(x,y)  \not\in E(G)$.
  \item $x$ is a \emph{true in-twin} ($\tit$) of $y$ if $(x,y) \in E(G)$, $(y,x) \not\in E(G)$.
  \item $x$ is a \emph{bioriented true twin} ($\tdt$) of $y$ if $(x,y), (y,x) \in E(G)$.
 \end{itemize}
 \item A vertex $v \in V(G)$ is called \emph{pendant} if $|N^+(v)|+|N^-(v)|=1$. We distinguish between
 \begin{itemize}
  \item $v$ is a \emph{pendant plus} vertex ($\ppv$)  if $|N^+(v)|=1$ and $|N^-(v)|=0$.
  \item $v$ \emph{pendant minus} vertex ($\pmv$) if $|N^+(v)|=0$ and $|N^-(v)|=1$.
 \end{itemize}
 \end{itemize}
\end{definition}

This leads to the definition of a recursively defined graph class which is close to the definition for undirected distance-hereditary graphs. We denote this class of digraphs as directed twin-distance-hereditary graphs.

\begin{definition}[directed twin-distance-hereditary graphs]\label{def:didh}
A digraph $G$ is \emph{directed twin-distance-hereditary}, twin-dh or in $\DDH$ for short, if it can be constructed recursively by taking disjoint union, adding twins and pendant vertices, starting from a single vertex.

A \emph{directed pruning sequence} for $G$ is a sequence $S=(s_1, \dots, s_{n-1})$, where $\sigma= (v_0, \dots, v_{n-1})$ is an ordering of $V(G)$ and every $s_i$ is one of the following triples:
\begin{itemize}
\item $(v_i, \ppv, v_{a_i})$ if $v_i$ is a pendant plus vertex of $v_{a_i}$ in $G[\{v_0,\ldots , v_i\}]$
\item $(v_i, \pmv, v_{a_i})$ if $v_i$ is a pendant minus vertex of $v_{a_i}$ in $G[\{v_0,\ldots , v_i\}]$
\item $(v_i, \ft, v_{a_i})$ if $v_i$ is a false twin of $v_{a_i}$ in $G[\{v_0,\ldots , v_i\}]$
\item $(v_i, \tot, v_{a_i})$ if $v_i$ is a true out-twin of $v_{a_i}$ in $G[\{v_0,\ldots , v_i\}]$
\item $(v_i, \tit, v_{a_i})$ if $v_i$ is a true in-twin of $v_{a_i}$ in $G[\{v_0,\ldots , v_i\}]$
\item $(v_i, \tdt, v_{a_i})$ if $v_i$ is a bioriented true twin of $v_{a_i}$ in $G[\{v_0,\ldots , v_i\}]$
\end{itemize}

In general, we denote $s_i = (v_i, op_i, v_{a_i})$ and say for vertex $v_i$, that $op_i$ is the \emph{operation} and $v_{a_i}$ the \emph{anchor vertex} of $s_i$.
\end{definition}

Like in the undirected case, for a given twin-dh digraph, it is easy to get a directed pruning sequence.

\begin{proposition}
Let $G$ be a twin-distance-hereditary digraph. Then, a directed pruning sequence of $G$ can be computed in polynomial time.
\end{proposition}

\section{Properties of Twin-DH Digraphs}

The class of directed twin-distance-hereditary graphs is closed under the connected induced subgraph operation.

\begin{lemma}\label{lem:closed}
Let $G$ be a twin-dh digraph and let $H$ be a weakly connected induced subdigraph of $G$. Then $H$ is a twin-dh digraph.
\end{lemma}

\begin{proof}
  Let $G\in \DDH$ with $V(G)=\{v_0, \dots, v_{n-1} \}$ and let $S(G)=(s_1, \dots, s_{n-1})$ with $\sigma(G)= (v_0, \dots, v_{n-1})$ be a directed pruning sequence of $G$.
  Let $H = G \setminus \{v\}$ be the weakly connected induced subdigraph $H$ of $G$ which emerges when deleting vertex $v$ and all corresponding edges from $G$.
  We then create a directed pruning sequence $S(H)$ with ordering $\sigma(H)$ with the following procedures for the three different cases.
  \begin{enumerate}
    \item If $v=v_0$, we just delete $s_1$ from $S(G)$ to obtain $S(H)$ and adjust the indies, now $v_1$ is the first vertex in $\sigma(H)$.
    \item If there exists $(v,op_i,a_i)\in S(G)$ and no $(u_j,op_j,v)$ with $i<j$:\\
    (After generating $v$ in $S(G)$, $v$ never occurs as an anchor vertex.)\\
    In this case we get $S(H)$ by deleting $(v,op_i,a_i)$ from $S(G)$ and adjust the indices.
    \item If there exists $(v,op_i,a_i)\in S(G)$ and also $ (u_{j1},op_{j1},v),...,(u_{jk},op_{jk},v)\in S(G)$ with $i<j$ and $k,j \leq m-1$:\\
    (After generating $v$ in $S(G)$, $v$ occurs at least once as an anchor vertex.)\\
    Since the emerging digraph must be weakly connected, it holds that $op_{jk}$ must be a directed twin operation.
    We get a temporary $S'(H)=(s'_1, \ldots s'_{m-2})$ by the following steps.
    \begin{itemize}
      \item
      For $t=1, \ldots, i-1$ let $s'_t=s_t$.
      We keep the directed pruning sequence until $v$ is generated.
      \item
      For $t=i$ we set $s'_i=(v',op_i, a_i)$ where $v'$ is the vertex such that $s_h=(v',op_h,v)$ with $op_h$ is a directed twin operation and $\not\exists s_p=(v'',op_p, v)$ with $p>h$ and $op_h$ is a directed twin operation.
      Thus, $v'$ is the last twin of $v$ with respect to $S(G)$.
      Now we replace $v$ by $v'$ as an anchor in all following occurrences.
      As $v'$ is the latest twin of $v$ w.r.t. $S(G)$, every operation applied on $v$ is also applied on $v'$.
      \item
      For $t=i+1,\ldots,\ell$ with $i+1 \leq \ell \leq m-1$ and $v=a_t$ first check if $v'=u_t$ in $s_t=(u_t,op_t,a_t)$.
      If this situation arrives, we delete this $s_t$ from our pruning sequence, such that we set $s'_t=(,,)$.
      Vertex $v'$ is now generated earlier in the directed pruning sequence and we do not need this step anymore.
      We will delete this empty triple at the very end, such that we don't have counting issues in the following procedure.
      As long as $v'\neq u_t$ we set $s'_t=s_t$ if $v\neq a_t$ and we set $s'_t=(u_t,op_t, v')$ if $v=a_t$ for $s_t=(u_t,op_t,a_t)$.
      \item For the remaining $t=\ell + 1 , \ldots, m-1$ we set $s'_t=s_t$.
    \end{itemize}
    At the end of this procedure, we delete the empty entry $s_c=(,,)$ from $S'(H)$, adjust the indices and get a directed pruning sequence $S(H)$ for $H$.
  \end{enumerate}
  This holds for every weakly connected subdigraph $H$, since we can repeat this procedure for every vertex which is in $G$ but not in $H$.
  Thus, we can always get a directed pruning sequence $S(H)$ and $H$ is a twin-dh digraph.
\end{proof}

As every directed pruning sequence can easily be transformed into a pruning sequence, the relation to undirected distance-hereditary graphs follows immediately.

\begin{proposition}
  If $G$ is a twin-dh digraph, then $un(G)$ is distance-hereditary.
\end{proposition}

\subsection{Sub- and Superclasses of Twin-DH Digraphs}

In the undirected case, distance-hereditary graphs can be classified into the hierarchy with other graph classes. Especially, they are a superclass of co-graphs by the definition of co-graphs using twins. We now show, that this is also possible in the directed case.

\begin{proposition}\label{pro:twins}
Every directed co-graph with at least two vertices has directed twins.
\end{proposition}
\begin{proof} Let $G$ be a directed co-graph with at least two vertices.
If $G$ has exactly two vertices, then these are twins.
So, let $G$ have more than two vertices.
Then $G=G_1\star G_2$ for some directed co-graphs $G_1$ and $G_2$ with $|V(G_1)|\ge 2$ or $|V(G_2)|\ge 2$, where $\star\in\{\oplus, \oslash, \otimes\}$.
By induction, $G_1$ or $G_2$ has twins $x,y$.
Now, by definition of the $\star$-operation, $x$ and $y$ are also twins in $G$.
Thus, every directed co-graph with at least two vertices has a twins as claimed.
\end{proof}

\begin{theorem}\label{thm:1}
A digraph is a directed co-graph if and only if it can be constructed recursively by taking disjoint union and adding directed twins, starting from a single vertex.
\end{theorem}

\begin{proof}
Note that we may assume that all graphs considered have at least two vertices. Otherwise, the theorem clearly holds.

First, let $G$ be a directed co-graph. Then, by Proposition~\ref{pro:twins}, $G$ has twins $x$ and~$y$. Let $G'=G-y$.
Since $G'$ is again a directed co-graph, by induction, $G'$ can be constructed by taking disjoint union and adding twins, starting from single vertices.
Since $G$ is obtained from $G'$ by adding twin $y$ to $x$, $G$ therefore can be constructed by taking disjoint union and adding twins, starting from single vertices, too.

For the other direction, suppose that $G$ can be constructed by taking disjoint union and adding twins, starting from single vertices.
We see by induction that $G$ is a directed co-graph.
Now, if $G$ is disconnected, then, as every component of $G$ is a directed co-graph, $G$ is a directed co-graph.
So, let us assume that $G$ is connected. As every digraph with at most two vertices is a directed co-graph, we may also assume that $G$ has more than two vertices. Now, by our assumption, $G$ has twins $x$ and $y$ so that $y$ is the last vertex adding to $G-y$ in obtaining $G$.
Let $G'=G-y$. Since $G'$ can be constructed by taking disjoint union and adding twins, $G'$ is a directed co-graph by induction.
Since $G'$ is connected and has at least two vertices, $G'=G_1'\star G_2'$ for some directed co-graphs $G_1'$ and $G_2'$, where $\star\in\{\oslash,\otimes\}$. Let $x\in G_1'$, say.
Write $G_1=G[V(G_1')\cup\{y\}]$ and $G_2=G_2'$, and note that $G_1$ and $G_2$ are directed co-graphs.

Then, since $x,y$ are twins in $G$, $G=G_1\star G_2$. Hence $G$ is a directed co-graph, and the proof of Theorem~\ref{thm:1} is complete.

\end{proof}

Then, the relation to twin-dh digraphs follows immediately:

\begin{corollary}
Let $G$ be a directed co-graph. Then, $G$ is also twin-distance-hereditary.
\end{corollary}

\subsubsection{Strong components in Twin-DH digraphs}

By Lemma \ref{lem:closed} and Theorem \ref{thm:1}, we can further conclude the following result:

\begin{lemma}\label{lem:twin_dh_co}
Let $G$ be a twin-dh digraph. Then every strong component of $G$ is a directed co-graph.
\end{lemma}

\begin{proof}
Let $H$ be an induced subdigraph of $G$ that is strongly connected. Then, by Lemma \ref{lem:closed}, $H$ is a twin-dh digraph. Thus, there is a directed pruning sequence $S(H)$, that creates $H$.
Assume that there is an element $s_i=(v_i, op_i, v_{a_i})$ in $S$ with operation $op_i$ is a pendant plus (respectively pendant minus) operation. Then, by the allowed operations in twin-dh digraphs, there is no directed path from $v_{a_i}$ to $v_i$ (respectively from $v_i$ to $v_{a_i}$) in $H$. This is a contradiction to the fact, that $H$ is strongly connected.
Thus, $S$ does not contain any pendant vertex operations. By Theorem \ref{thm:1} follows, that $H$ is a directed co-graph.
\end{proof}

This lemma admits many algorithmic results. Every digraph problem, which is solvable by considering only the strong components and which is further computable on directed co-graphs, is similarly computable on twin-dh digraphs by Lemma \ref{lem:twin_dh_co}. For example, this holds for several directed graph parameters, as we see later on.

With these results it also possible to show that twin-dh digraphs are a subclass of extended directed co-graphs:

\begin{proposition}\label{cor_extended}
Let $G$ be a twin-dh digraph. Then $G$ is also an extended directed co-graph.
\end{proposition}

\begin{proof}
  Let $G$ be a twin-dh digraph. With the following procedure we can get a construction of $G$ with the extended directed co-graph operations. We know from Lemma \ref{lem:twin_dh_co} that the strong components are directed co-graphs, thus we build the di-co-tree of these components. If a vertex does not belong to any bigger strong component it can be seen as its own strong component. The missing arc which connect the different strong components in $G$ are built by directed union operations, where we can leave our all arcs except for the arc of the corresponding pendant vertex.

\end{proof}

This result allows us to deduce some results how to solve several graph parameters on this graph class. However, we show that we can even do better on twin-dh digraphs. Furthermore, twin-dh digraphs have a decisive advantage compared to its superclass since it has bounded directed clique-width, which will be defined later. As we can build directed grids with the directed union operation in extended directed co-graphs, the directed clique-width for this class is not bounded. This allows us to solve many problems on twin-dh-digraphs, which cannot be solved on extended directed co-graphs. The reason for this is, that we lose information about the edges and therefore about the reachability within extended directed co-graphs which we preserve in the subclass.

Next, we show how this class is related to the class of distance-hereditary digraphs from \cite{LR10}.

\subsubsection{Twin-DH Digraphs are distance-hereditary}

Though for the definition we used the approach of a recursive construction by twins and pendant vertices, twin-dh digraphs still fulfill the distance-heredity property.

\begin{theorem} \label{ddh-ddh}
Every twin-distance-hereditary digraph $G$ is distance-hereditary, i.e. for every two vertices $u$ and $v$ in $V(G)$, all induced $u,v$-paths have the same length.
\end{theorem}

In \cite{LR10} the authors claim that for pendant vertices, for (slightly different, but more general defined) oriented twins and for false twins the distance-hereditary property remains fulfilled. However, the result that every path between two distinct vertices is of length one does not hold in general when including bioriented edges. This is why we need the following lemma, which leads us directly to the theorem above.

\begin{lemma}\label{lem:distance_of_twins}
 For two twins $u,v$ in a twin-dh digraph $G$ it holds that if there exists a path from $u$ to $v$ then the shortest path in every induced subdigraph $G'$ of $G$ is $\leq 2$.
\end{lemma}

Note that the proof could be shortened using Theorem 4 of \cite{LR10}.
\begin{proof}
  Let $u,v\in V(G)$ be twins in $G$.
  If they are bioriented twins, the distance between them is trivially $1$.
  If $u,v$ are oriented twins let w.l.o.g. be $(u,v)\in E(G)$. Then the distance from $u$ to $v$ is also $1$, but this is not the case for the other direction.
  So let $u$ and $v$ be oriented twins with $(v,u)\in E(G)$ or false twins. In order to proof the lemma by contradiction, we assume that there is a shortest path from $u$ to $v$ of length $\geq 3$ in an induced subdigraph $G'$ of $G$. Let this path be $P=(u,v_1,\ldots, v_k,v)$. Since $N^-_{G'}(v)=N^-_{G'}(u)$ and $N^+_{G'}(v)=N^+_{G'}(u)$ it holds that $(v,v_1)\in E(G')$ and $(v_k,u)\in E(G')$.
  Then there is a cycle $(u,v_1,\ldots, v_k,u)$ of length at least $3$.
  If the length is $3$ with there must be at least two bidirectional edges in this cycle, otherwise this cycle is not constructible by directed twins.
  But then, one of the bidirectional edges goes to $u$ and since $v$ is a twin, we could have taken this shorter path $(u,v_i,v)$ of length $2$ from the beginning, which is a contradiction to the assumption of length $3$.
  Let's assume the shortest path is $>3$. Then there is a cycle $u,v_1,v_2, \ldots, v_k,u$ with the same argumentation as before. Since $un(G)$ is distance-hereditary, there cannot be any holes, thus induced cycles of length $\geq 5$. Thus, the cycles must contain edges in between. If these edges are forward edges along the cycle they would shorten the path from $u$ to $v$ which is a contradiction. If these edges are backward edges along the cycle, they would again build smaller induced cycles, up to a $\overrightarrow{C_3}$ which is not constructible by a directed pruning sequence. Backward edges are only possible, if the outer edges from the cycle are bioriented. But this would build an induced subdigraph $H_{19}$ or $H_{16}$ (Fig. \ref{fig:forbidden_DH}), which are not constructible with a directed pruning sequence and thus are not directed twin-distance-hereditary.
  Thus, such a path cannot exists and the shortest path is always of length $\leq 2$.
\end{proof}

With the same example as for extended co-graphs, the class of distance-hereditary digraphs has unbounded directed clique-width. In a grid digraph, where all edges are directed from the top to the bottom and left to right, the directed clique-width increases with the number of vertices. Here we see a certain advantage of the class of  twin-dh digraphs which justifies to take a closer look.

\subsection{Characterization of Twin-DH Digraphs}

As already mentioned previously, our definition is not only based on the property of distance heredity as in the undirected case, or regarding distance-hereditary digraphs.
That is, not every digraph which is distance-hereditary, is also a twin-dh digraph. This can be easily shown by e.g. a bioriented path.
However, it is possible to give different characterizations of the class $\DDH$ by forbidden induced subdigraphs.

We give a characterization by forbidden induced subdigraphs.
Therefore, we first need to define the two-leaves-digraph.

\begin{definition}\label{def:two-leaves}
  A weakly connected digraph $G$ is a two-leaves-digraph if it has at least $4$ vertices and if it contains at least two bioriented leaves $u,v$ with $N(u)\neq N(v)$ in $un(G)$, see Fig. \ref{fig:forbidden_DH}.
\end{definition}

\begin{theorem}\label{the:forbiddenSG}
  A digraph $G$ is directed Twin-distance-hereditary if and only if it contains none of the the following graphs, see Fig. \ref{fig:forbidden_DH} as induced subdigraph.
  \begin{itemize}
    \item $\overrightarrow{C_3}$.
    \item any biorientation of the $C_n$ (hole) for $n\geq 5$, domino, house or gem.
    \item $H_0 , \ldots , H_{27}$.
    \item A two-leaves-digraph.
  \end{itemize}
\end{theorem}

\begin{proof}
\begin{itemize}
  \item $\Rightarrow$ None of the graphs can be constructed with directed twins and directed pendant vertices and the class is hereditary, see Lemma \ref{lem:closed}.
  \item $\Leftarrow$ We proof this by contradiction. Let $G$ be a graph that does not contain any of the forbidden induced subgraphs above and let's assume that $G\not\in \DDH$. A graph is not Twin-distance-hereditary if it cannot be constructed by the directed twin and pendant vertices operations. We distinguish two cases $G\not\in \DDH \land un(G)\not\in \UDH$ and $G\not\in \DDH \land un(G)\in DH$. Case one is that $G$ is not twin-dh for structural reasons, thus $G\not\in \DDH \land un(G)\not\in \UDH$ which is ensured by the exclusion of any biorientation of the $C_n$ (hole) for
  $n\geq 5$, domino, house or gem, see Fig. \ref{fig:forbidden_DH}.\\
  In the other case $G\not\in \DDH \land un(G)\in \UDH$ the graph is not twin-dh for orientation reasons. This means that there exists a pruning sequence $P$ for $un(G)$ but there is no directed pruning sequence for $G$ because the arcs have a biorientation, which cannot be achieved by the directed twin and pendant vertex operations. By forbidding the two-leaves-digraphs, the pendant vertex operations are not allowed to be bioriented and thus, every undirected pendant vertex operation in $P$ can be replaced by a directed pendant vertex operation.
  It is left to show that $G$ has as well none of the digraphs of set $\mathcal{H}=\{\overrightarrow{C_3}, H_0,\ldots , H_{27}\}$ as induced subdigraph. Note that $\mathcal{H}$ contains every graph with $\leq 4$ vertices that cannot be constructed by the directed twin operations, with no inclusions. If we look at every possibly biorientation of the operations in $P$, we get any possible directed pruning sequence of $G$. For every induced subdigraph $H$ of $G$ with $\leq 4$ vertices there must exists a biorientation, such that there is a directed pruning sequence, since the set $\mathcal{H}$ is exactly the set of graphs with $\leq 4$ that has no directed pruning sequence.
  There are no more forbidden induced subdigraphs $H'$ that contains none of the previous excluded graphs as induced subgraph with more than $4$ vertices for the following reason.
  Assume there is an induced subdigraph $H'$ of $G$ with $\geq 5$ vertices which is minimal in the sense that it does not contain a graph from $\mathcal{H}$ as induced subdigraph and for which there is no directed pruning sequence.
  Let $V(H')=\{t_1,t_2,u,v,w_1,...,w_k\}$ with $k \geq 1$ be the vertex set of $H'$, where $t_1$ and $t_2$ are twins in the undirected pruning sequence $P'$ of $H'$.
  As $H'$ is minimal, every induced subdigraph $H^*$ of $H'$ with $4$ vertices is not forbidden. Thus, the different directed neighborhood of $t_1$ and $t_2$ must arise by adding the fifth vertex $w_1$.
  But if this vertex causes an orientation problem in $H[\{t_1,t_2,u,v,w_1\}]$ then this vertex also causes an orientation problem in $H[\{t_1,t_2,u,w_1\}]$ which would build a forbidden induced subdigraph with $4$ vertices. We end up in the same problem if we chose any other two twins. Thus, there cannot be a forbidden induced subdigraph with more than $4$ vertices for which there is no directed pruning sequence, if an undirected pruning sequence exists and $G\in \DDH$.
\end{itemize}

\end{proof}

To get an better understanding of the construction of the forbidden induced subdigraphs $H_0 , \ldots , H_{27}$ we group them as follows. In none of them we can find directed twins, but the undirected versions of them are distance hereditary.
\begin{itemize}
  \item $H_0,\ldots , H_5$: Digraphs with $4$ or less vertices with $un(G)=C_4$ which are strongly connected.
  \item $H_6,\ldots , H_9$: Digraphs with $4$ vertices with $un(G)=C_4$ which are not strongly connected.
  \item $H_{10},\ldots , H_{19}$: Digraphs with $4$ vertices with $un(G)=C_4$ with an additional single diagonal edge.
  \item $H_{20},\ldots , H_{26}$: Digraphs with $4$ vertices with $un(G)=C_4$ with an additional bidirectional diagonal edge.
  \item $H_{27}$: Forbidden orientation of the $K_4$.
\end{itemize}

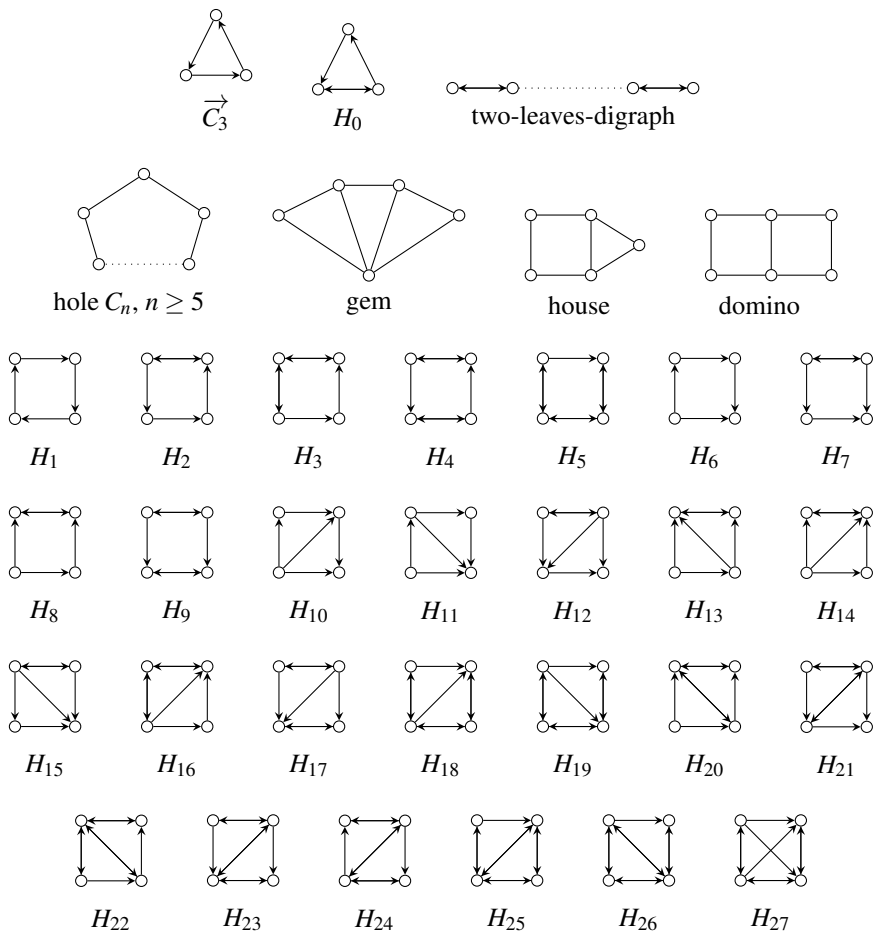
\begin{figure}[h!]
  \tikzstyle{vertexVi}=[circle,inner sep=1.5pt,fill=black]
  \tikzstyle{vertexY}=[draw,circle,inner sep=1.5pt]
  \tikzstyle{vertex}=[draw,circle,inner sep=1.5pt]
  \begin{center}
  \begin{tikzpicture}[>=stealth, scale=.4]
  \node[vertex] (x) at (0,1) {};
  \node[vertex] (y) at (2,1) {};
  \node[vertex] (z) at (1,3) {};
  \draw[->] (x)--(y);
  \draw[->] (y)--(z);
  \draw[->] (z)--(x);
  \node at (1,-.3) {$\overrightarrow{C_3}$};
  \end{tikzpicture}
  \qquad
  \begin{tikzpicture}[>=stealth, scale=.4]
  \node[vertex] (x1) at (0,0) {};
  \node[vertex] (x2) at (2,0) {};
  \node[vertex] (x3) at (1,2) {};
  \draw[->] (x1)--(x2);
  \draw[->] (x2)--(x3);
  \draw[->] (x2)--(x1);
  \draw[->] (x3)--(x1);
  \node at (1,-1) {$H_0$};
  \end{tikzpicture}
  \qquad
  \begin{tikzpicture}[>=stealth, scale=.4]
  \node[vertex] (x) at (0,0) {};
  \node[vertex] (y) at (2,0) {};
  \node[vertex] (z) at (6,0) {};
  \node[vertex] (w) at (8,0) {};
  \draw[->] (x)--(y);
  \draw[->] (y)--(x);
  \draw[->] (z)--(w);
  \draw[->] (w)--(z);
  \draw[-, dotted] (y)--(z);
  \node at (0,-.6) {};
  \node at (2,-.6) {};
  \node at (6,-.6) {};
  \node at (8,-.6) {};
  \node at (4,-1) {two-leaves-digraph};
  \end{tikzpicture}
  \end{center}

  %1
  \begin{center}
  \begin{tikzpicture}[>=stealth, scale=.4]
  \node[vertex] (x) at (0.5,1) {};
  \node[vertex] (y) at (0,2.7) {};
  \node[vertex] (z) at (2,4) {};
  \node[vertex] (v) at (4,2.7) {};
  \node[vertex] (w) at (3.5,1) {};
  \draw[-] (x)--(y);
  \draw[-] (y)--(z);
  \draw[-] (z)--(v);
  \draw[-] (w)--(v);
  \draw[dotted] (w)--(x);
  \node at (1.5,-.3) {hole $C_n$, $n\geq 5$};
  \end{tikzpicture}
  \qquad
  \begin{tikzpicture}[>=stealth, scale=.4] \node[vertex] (a) at (0,2) {};
  \node[vertex] (b) at (2,3) {};
  \node[vertex] (c) at (4,3) {};
  \node[vertex] (d) at (6,2) {};
  \node[vertex] (u) at (3,0) {};
  \draw[-] (d)--(c);
  \draw[-] (c)--(b);
  \draw[-] (b)--(a);
  \draw[-] (c)--(u);
  \draw[-] (d)--(u);
  \draw[-] (u)--(a);
  \draw[-] (u)--(b);
  \node at (3,-1) {gem};
  \end{tikzpicture}
  \qquad
  \begin{tikzpicture}[>=stealth, scale=.4] \node[vertex] (x1) at (0,0) {};
  \node[vertex] (x2) at (0,2) {};
  \node[vertex] (x3) at (2,2) {};
  \node[vertex] (x4) at (2,0) {};
  \node[vertex] (x5) at (3.6,1) {};
  \draw[-] (x1)--(x2);
  \draw[-] (x1)--(x4);
  \draw[-] (x2)--(x3);
  \draw[-] (x3)--(x4);
  \draw[-] (x3)--(x5);
  \draw[-] (x4)--(x5);
  \node at (1.6,-1) {house};
  \end{tikzpicture}
  \qquad
  \begin{tikzpicture}[>=stealth, scale=.4] \node[vertex] (x1) at (0,0) {};
  \node[vertex] (x2) at (0,2) {};
  \node[vertex] (x3) at (2,2) {};
  \node[vertex] (x4) at (2,0) {};
  \node[vertex] (x5) at (4,0) {};
  \node[vertex] (x6) at (4,2) {};
  \draw[-] (x1)--(x2);
  \draw[-] (x1)--(x4);
  \draw[-] (x2)--(x3);
  \draw[-] (x3)--(x4);
  \draw[-] (x3)--(x6);
  \draw[-] (x4)--(x5);
  \draw[-] (x5)--(x6);
  \node at (1.6,-1) {domino};
  \end{tikzpicture}
  \ifdefined\journal
  \qquad
  \begin{tikzpicture}[>=stealth, scale=.4]
  \node[vertex] (x1) at (0,0) {};
  \node[vertex] (x2) at (2,0) {};
  \node[vertex] (x3) at (1,2) {};
  \draw[-] (x1)--(x2);
  \draw[-] (x1)--(x3);
  \draw[-] (x2)--(x3);
  \node at (1,-1) {$triangle$};
  \end{tikzpicture}
  \fi
  \end{center}

  %2
  \begin{center}
  \begin{tikzpicture}[>=stealth, scale=.4]
  \node[vertex] (x) at (0,1) {};
  \node[vertex] (y) at (0,3) {};
  \node[vertex] (z) at (2,3) {};
  \node[vertex] (w) at (2,1) {};
  \draw[->] (x)--(y);
  \draw[->] (y)--(z);
  \draw[->] (z)--(w);
  \draw[->] (w)--(x);
  \node at (1,-.3) {$H_1$};
  \end{tikzpicture}
  \qquad
  \begin{tikzpicture}[>=stealth, scale=.4]
  \node[vertex] (x) at (0,1) {};
  \node[vertex] (y) at (2,1) {};
  \node[vertex] (z) at (0,3) {};
  \node[vertex] (w) at (2,3) {};
  \draw[->] (z)--(w);
  \draw[->] (w)--(z);
  \draw[->] (x)--(y);
  \draw[->] (z)--(x);
  \draw[->] (y)--(w);
  \node at (1,-.3) {$H_2$};
  \end{tikzpicture}
  \qquad
  \begin{tikzpicture}[>=stealth, scale=.4]
  \node[vertex] (x) at (0,1) {};
  \node[vertex] (y) at (2,1) {};
  \node[vertex] (z) at (0,3) {};
  \node[vertex] (w) at (2,3) {};
  \draw[->] (z)--(w);
  \draw[->] (x)--(z);
  \draw[->] (w)--(z);
  \draw[->] (z)--(x);
  \draw[->] (x)--(y);
  \draw[->] (y)--(w);
  \node at (1,-.3) {$H_3$};
  \end{tikzpicture}
  \qquad
  \begin{tikzpicture}[>=stealth, scale=.4]
  \node[vertex] (x) at (0,1) {};
  \node[vertex] (y) at (2,1) {};
  \node[vertex] (z) at (0,3) {};
  \node[vertex] (w) at (2,3) {};
  \draw[->] (z)--(w);
  \draw[->] (w)--(z);
  \draw[->] (x)--(y);
  \draw[->] (y)--(x);
  \draw[->] (z)--(x);
  \draw[->] (y)--(w);
  \node at (1,-.3) {$H_4$};
  \end{tikzpicture}
  \qquad
  \begin{tikzpicture}[>=stealth, scale=.4]
  \node[vertex] (x) at (0,1) {};
  \node[vertex] (y) at (2,1) {};
  \node[vertex] (z) at (0,3) {};
  \node[vertex] (w) at (2,3) {};
  \draw[->] (z)--(w);
  \draw[->] (z)--(x);
  \draw[->] (x)--(z);
  \draw[->] (x)--(y);
  \draw[->] (y)--(x);
  \draw[->] (y)--(w);
  \draw[->] (w)--(y);
  \node at (1,-.3) {$H_5$};
  \end{tikzpicture}
  \qquad
  \begin{tikzpicture}[>=stealth, scale=.4]
  \node[vertex] (x) at (0,1) {};
  \node[vertex] (y) at (2,1) {};
  \node[vertex] (z) at (0,3) {};
  \node[vertex] (w) at (2,3) {};
  \draw[->] (x)--(y);
  \draw[->] (x)--(z);
  \draw[->] (z)--(w);
  \draw[->] (w)--(y);
  \node at (1,-.3) {$H_6$};
  \end{tikzpicture}
  \qquad
  \begin{tikzpicture}[>=stealth, scale=.4]
  \node[vertex] (x) at (0,1) {};
  \node[vertex] (y) at (2,1) {};
  \node[vertex] (z) at (0,3) {};
  \node[vertex] (w) at (2,3) {};
  \draw[->] (z)--(x);
  \draw[->] (w)--(y);
  \draw[->] (z)--(w);
  \draw[->] (w)--(z);
  \draw[->] (x)--(y);
  \node at (1,-.3) {$H_7$};
  \end{tikzpicture}
  \end{center}

  %3
  \begin{center}
  \begin{tikzpicture}[>=stealth, scale=.4]
  \node[vertex] (x) at (0,1) {};
  \node[vertex] (y) at (2,1) {};
  \node[vertex] (z) at (0,3) {};
  \node[vertex] (w) at (2,3) {};
  \draw[->] (x)--(z);
  \draw[->] (y)--(w);
  \draw[->] (z)--(w);
  \draw[->] (w)--(z);
  \draw[->] (x)--(y);
  \node at (1,-.3) {$H_8$};
  \end{tikzpicture}
  \qquad
  \begin{tikzpicture}[>=stealth, scale=.4]
  \node[vertex] (x) at (0,1) {};
  \node[vertex] (y) at (2,1) {};
  \node[vertex] (z) at (0,3) {};
  \node[vertex] (w) at (2,3) {};
  \draw[->] (z)--(w);
  \draw[->] (w)--(z);
  \draw[->] (x)--(y);
  \draw[->] (y)--(x);
  \draw[->] (z)--(x);
  \draw[->] (w)--(y);
  \node at (1,-.3) {$H_9$};
  \end{tikzpicture}
  \qquad
  \begin{tikzpicture}[>=stealth, scale=.4]
  \node[vertex] (x) at (0,1) {};
  \node[vertex] (y) at (2,1) {};
  \node[vertex] (z) at (0,3) {};
  \node[vertex] (w) at (2,3) {};
  \draw[->] (x)--(y);
  \draw[->] (x)--(z);
  \draw[->] (z)--(w);
  \draw[->] (w)--(y);
  \draw[->] (x)--(w);
  \node at (1,-.3) {$H_{10}$};
  \end{tikzpicture}
  \qquad
  \begin{tikzpicture}[>=stealth, scale=.4]
  \node[vertex] (x) at (0,1) {};
  \node[vertex] (y) at (2,1) {};
  \node[vertex] (z) at (0,3) {};
  \node[vertex] (w) at (2,3) {};
  \draw[->] (x)--(y);
  \draw[->] (x)--(z);
  \draw[->] (z)--(w);
  \draw[->] (w)--(y);
  \draw[->] (z)--(y);
  \node at (1,-.3) {$H_{11}$};
  \end{tikzpicture}
  \qquad
  \begin{tikzpicture}[>=stealth, scale=.4]
  \node[vertex] (x) at (0,1) {};
  \node[vertex] (y) at (2,1) {};
  \node[vertex] (z) at (0,3) {};
  \node[vertex] (w) at (2,3) {};
  \draw[->] (w)--(z);
  \draw[->] (z)--(w);
  \draw[->] (z)--(x);
  \draw[->] (w)--(y);
  \draw[->] (w)--(x);
  \draw[->] (x)--(y);
  \node at (1,-.3) {$H_{12}$};
  \end{tikzpicture}
  \qquad
  \begin{tikzpicture}[>=stealth, scale=.4]
  \node[vertex] (x) at (0,1) {};
  \node[vertex] (y) at (2,1) {};
  \node[vertex] (z) at (0,3) {};
  \node[vertex] (w) at (2,3) {};
  \draw[->] (w)--(z);
  \draw[->] (x)--(y);
  \draw[->] (x)--(z);
  \draw[->] (y)--(z);
  \draw[->] (y)--(w);
  \draw[->] (z)--(w);
  \node at (1,-.3) {$H_{13}$};
  \end{tikzpicture}
  \qquad
  \begin{tikzpicture}[>=stealth, scale=.4]
  \node[vertex] (x) at (0,1) {};
  \node[vertex] (y) at (2,1) {};
  \node[vertex] (z) at (0,3) {};
  \node[vertex] (w) at (2,3) {};
  \draw[->] (w)--(z);
  \draw[->] (x)--(y);
  \draw[->] (x)--(z);
  \draw[->] (x)--(w);
  \draw[->] (y)--(w);
  \draw[->] (z)--(w);
  \node at (1,-.3) {$H_{14}$};
  \end{tikzpicture}
  \end{center}

  \begin{center}
  \begin{tikzpicture}[>=stealth, scale=.4]
  \node[vertex] (x) at (0,1) {};
  \node[vertex] (y) at (2,1) {};
  \node[vertex] (z) at (0,3) {};
  \node[vertex] (w) at (2,3) {};
  \draw[->] (w)--(z);
  \draw[->] (z)--(w);
  \draw[->] (z)--(x);
  \draw[->] (w)--(y);
  \draw[->] (x)--(y);
  \draw[->] (z)--(y);
  \node at (1,-.3) {$H_{15}$};
  \end{tikzpicture}
  \qquad
  \begin{tikzpicture}[>=stealth, scale=.4]
  \node[vertex] (x) at (0,1) {};
  \node[vertex] (y) at (2,1) {};
  \node[vertex] (z) at (0,3) {};
  \node[vertex] (w) at (2,3) {};
  \draw[->] (z)--(w);
  \draw[->] (w)--(z);
  \draw[->] (z)--(x);
  \draw[->] (x)--(z);
  \draw[->] (x)--(y);
  \draw[->] (y)--(w);
  \draw[->] (x)--(w);
  \node at (1,-.3) {$H_{16}$};
  \end{tikzpicture}
  \qquad
  \begin{tikzpicture}[>=stealth, scale=.4]
  \node[vertex] (x) at (0,1) {};
  \node[vertex] (y) at (2,1) {};
  \node[vertex] (z) at (0,3) {};
  \node[vertex] (w) at (2,3) {};
  \draw[->] (z)--(w);
  \draw[->] (w)--(z);
  \draw[->] (x)--(y);
  \draw[->] (y)--(x);
  \draw[->] (z)--(x);
  \draw[->] (w)--(y);
  \draw[->] (w)--(x);
  \node at (1,-.3) {$H_{17}$};
  \end{tikzpicture}
  \qquad
  \begin{tikzpicture}[>=stealth, scale=.4]
  \node[vertex] (x) at (0,1) {};
  \node[vertex] (y) at (2,1) {};
  \node[vertex] (z) at (0,3) {};
  \node[vertex] (w) at (2,3) {};
  \draw[->] (x)--(z);
  \draw[->] (z)--(x);
  \draw[->] (x)--(y);
  \draw[->] (y)--(x);
  \draw[->] (w)--(y);
  \draw[->] (y)--(w);
  \draw[->] (x)--(w);
  \draw[->] (z)--(w);
  \node at (1,-.3) {$H_{18}$};
  \end{tikzpicture}
  \qquad
  \begin{tikzpicture}[>=stealth, scale=.4]
  \node[vertex] (x) at (0,1) {};
  \node[vertex] (y) at (2,1) {};
  \node[vertex] (z) at (0,3) {};
  \node[vertex] (w) at (2,3) {};
  \draw[->] (x)--(z);
  \draw[->] (z)--(x);
  \draw[->] (x)--(y);
  \draw[->] (y)--(x);
  \draw[->] (w)--(y);
  \draw[->] (y)--(w);
  \draw[->] (z)--(w);
  \draw[->] (z)--(y);
  \node at (1,-.3) {$H_{19}$};
  \end{tikzpicture}
  \qquad
  \begin{tikzpicture}[>=stealth, scale=.4]
  \node[vertex] (x) at (0,1) {};
  \node[vertex] (y) at (2,1) {};
  \node[vertex] (z) at (0,3) {};
  \node[vertex] (w) at (2,3) {};
  \draw[->] (w)--(z);
  \draw[->] (x)--(y);
  \draw[->] (x)--(z);
  \draw[->] (y)--(z);
  \draw[->] (y)--(w);
  \draw[->] (z)--(w);
  \draw[->] (z)--(y);
  \node at (1,-.3) {$H_{20}$};
  \end{tikzpicture}
  \qquad
  \begin{tikzpicture}[>=stealth, scale=.4]
  \node[vertex] (x) at (0,1) {};
  \node[vertex] (y) at (2,1) {};
  \node[vertex] (z) at (0,3) {};
  \node[vertex] (w) at (2,3) {};
  \draw[->] (w)--(z);
  \draw[->] (z)--(w);
  \draw[->] (z)--(x);
  \draw[->] (w)--(y);
  \draw[->] (w)--(x);
  \draw[->] (x)--(y);
  \draw[->] (x)--(w);
  \node at (1,-.3) {$H_{21}$};
  \end{tikzpicture}
  \end{center}

  %5
  \begin{center}
  \begin{tikzpicture}[>=stealth, scale=.4]
  \node[vertex] (x) at (0,1) {};
  \node[vertex] (y) at (2,1) {};
  \node[vertex] (z) at (0,3) {};
  \node[vertex] (w) at (2,3) {};
  \draw[->] (z)--(w);
  \draw[->] (w)--(z);
  \draw[->] (z)--(x);
  \draw[->] (x)--(z);
  \draw[->] (x)--(y);
  \draw[->] (y)--(w);
  \draw[->] (y)--(z);
  \draw[->] (z)--(y);
  \node at (1,-.3) {$H_{22}$};
  \end{tikzpicture}
  \qquad
  \begin{tikzpicture}[>=stealth, scale=.4]
  \node[vertex] (x) at (0,1) {};
  \node[vertex] (y) at (2,1) {};
  \node[vertex] (z) at (0,3) {};
  \node[vertex] (w) at (2,3) {};
  \draw[->] (z)--(x);
  \draw[->] (x)--(y);
  \draw[->] (y)--(x);
  \draw[->] (w)--(y);
  \draw[->] (x)--(w);
  \draw[->] (z)--(w);
  \draw[->] (w)--(z);
  \draw[->] (w)--(x);
  \node at (1,-.3) {$H_{23}$};
  \end{tikzpicture}
  \qquad
  \begin{tikzpicture}[>=stealth, scale=.4]
  \node[vertex] (x) at (0,1) {};
  \node[vertex] (y) at (2,1) {};
  \node[vertex] (z) at (0,3) {};
  \node[vertex] (w) at (2,3) {};
  \draw[->] (x)--(z);
  \draw[->] (x)--(y);
  \draw[->] (y)--(x);
  \draw[->] (w)--(y);
  \draw[->] (x)--(w);
  \draw[->] (z)--(w);
  \draw[->] (w)--(z);
  \draw[->] (w)--(x);
  \node at (1,-.3) {$H_{24}$};
  \end{tikzpicture}
  \qquad
  \begin{tikzpicture}[>=stealth, scale=.4]
  \node[vertex] (x) at (0,1) {};
  \node[vertex] (y) at (2,1) {};
  \node[vertex] (z) at (0,3) {};
  \node[vertex] (w) at (2,3) {};
  \draw[->] (x)--(z);
  \draw[->] (z)--(x);
  \draw[->] (x)--(y);
  \draw[->] (y)--(x);
  \draw[->] (w)--(y);
  \draw[->] (y)--(w);
  \draw[->] (x)--(w);
  \draw[->] (w)--(x);
  \draw[->] (z)--(w);
  \node at (1,-.3) {$H_{25}$};
  \end{tikzpicture}
  \qquad
  \begin{tikzpicture}[>=stealth, scale=.4]
  \node[vertex] (x) at (0,1) {};
  \node[vertex] (y) at (2,1) {};
  \node[vertex] (z) at (0,3) {};
  \node[vertex] (w) at (2,3) {};
  \draw[->] (x)--(z);
  \draw[->] (y)--(z);
  \draw[->] (z)--(x);
  \draw[->] (x)--(y);
  \draw[->] (y)--(x);
  \draw[->] (w)--(y);
  \draw[->] (y)--(w);
  \draw[->] (z)--(w);
  \draw[->] (z)--(y);
  \node at (1,-.3) {$H_{26}$};
  \end{tikzpicture}
  \qquad
  \begin{tikzpicture}[>=stealth, scale=.4]
  \node[vertex] (x) at (0,1) {};
  \node[vertex] (y) at (2,1) {};
  \node[vertex] (z) at (0,3) {};
  \node[vertex] (w) at (2,3) {};
  \draw[->] (x)--(z);
  \draw[->] (z)--(x);
  \draw[->] (x)--(y);
  \draw[->] (y)--(x);
  \draw[->] (w)--(y);
  \draw[->] (y)--(w);
  \draw[->] (x)--(w);
  \draw[->] (z)--(w);
  \draw[->] (z)--(y);
  \node at (1,-.3) {$H_{27}$};
  \end{tikzpicture}
  \end{center}
  \caption{Forbidden induced sub(di)graphs.}
  \label{fig:forbidden_DH}
\end{figure}

\section{Directed Graph Parameters on Twin-DH Digraphs}

In \cite{GKR21a} we presented algorithms to compute different directed width measures on (extended) directed co-graphs in linear time. Among these are directed path-width, directed tree-width, DAG-width and cycle rank (see \cite{GKR21a} for formal definitions). Those algorithms are not extendable directly to twin-dh digraphs, but by Lemma \ref{lem:twin_dh_co}, the results can be expanded to the latter.

We have stated the following Lemma in \cite{GR19} for directed path-width and directed tree-width. The proof is extendable straight-forward to DAG-width and cycle rank.

\begin{lemma} \label{lem:pw-sc}
The directed path-width (directed tree-width, DAG-width and cycle rank respectively) of a digraph $G$ is the maximum of the directed path-widths (directed tree-widths, DAG-widths and cycle ranks respectively) of all strong components of $G$.
\end{lemma}

By this lemma, we can show that it is possible to bound the computation of the mentioned parameters on a twin-distance-hereditary digraph.

\begin{theorem}\label{main-theorem}
Let $G$ be a twin-distance-hereditary digraph, $n=|V(G)|$, $m=|E(G)|$. Then it holds that directed path-width ($\dpws$), directed tree-width ($\dtws$), DAG-width ($\dagws$), and cycle rank ($\cyr$) are computable in time $\mathcal{O}(n+m)$ and further
\begin{equation}\label{equ_parameter}
  \dpws(G)=\dtws(G)=\dagws(G)-1=\cyr(G).
\end{equation}
\end{theorem}

\begin{proof}
It is possible to get all strong components $C_1, \dots, C_r$ of $G$ in linear time. By Lemma \ref{lem:twin_dh_co}, all $C_i$, $1 \leq i \leq r$ are directed co-graphs. By \cite{GKR21a}, it is possible to get the directed path-width, directed tree-width, DAG-width and cycle rank of directed co-graphs in linear time and it holds that  $\dpws(C_i)=\dtws(C_i)=\dagws(C_i)-1=\cyr(C_i)$ for all $1 \leq i \leq r$. By Lemma \ref{lem:pw-sc}, the directed path width (directed tree-width, DAG-width and cycle rank respectively) of $G$ is the maximum of the directed path-widths (directed tree-widths, DAG-widths and cycle ranks respectively) over all $C_i$, $1 \leq i \leq r$.
It then follows that those parameters can be computed in linear time and that
$\dpws(G)=\dtws(G)=\dagws(G)-1=\cyr(G)$.
\end{proof}

Note that, as twin-dh digraphs are a subclass of extended directed co-graphs, the equality of these graph parameters follows directly from the results in \cite{GKR21a}. However, on extended directed co-graphs, there are only known algorithms to compute them in polynomial, not linear time.

\subsection{Directed Clique-width of Twin-DH Digraphs}

We now define the previously mentioned parameter directed clique-width. It differs from the formerly mentioned parameters as instead of representing the size of strong components in some way, it describes the number of different neighborhoods. Especially for bioriented cliques, the above mentioned parameters are infinitely large whereas directed clique-width is linear.

\begin{definition}[Directed clique-width \cite{CO00}]\label{Def_dcw}
The {\em directed clique-width} of a digraph $G$,
is the minimum number of labels that we need to define $G$ using the following operations:
\begin{enumerate}
\item Create a new vertex with label $a$ (denoted by $\bullet_a$).
\item Disjoint union of two labeled digraphs $G$ and $H$, denoted with $G\oplus H$.
\item Insert an arc from every vertex labeled with $a$ to every vertex labeled with $b$
 with $a\neq b$, which is denoted with $\alpha_{a,b}$.
\item Change label $a$ into label $b$, which is denoted by $\rho_{a\to b}$.
\end{enumerate}
An expression $X$ that is built using the operations defined above
with $k$ labels is a {\em directed clique-width $k$-expression}, we say {\em $k$-expression} for short. This structure is useful in solving many hard problems in polynomial time on digraphs with bounded directed clique-width, see e.g. \cite{GHKLOR14}.
\end{definition}

\begin{example}
A $3$-expression for the path on three vertices which is directed in one direction, namely $\overrightarrow{P_3}$ with vertices $a,b,c$, is
$$\alpha_{2,3}(\alpha_{1,2}(\bullet_1 \oplus \bullet_2 )\oplus \bullet_3 ).$$
\end{example}

Undirected clique-with was introduced by \cite{CO00} and is defined correspondingly. In the undirected case, co-graphs are exactly the graphs of clique-width at most $2$ and distance-hereditary graphs have clique-width at most $3$. This leads to the idea of regarding directed clique-width on twin-dh digraphs.

\begin{theorem}\label{the:cw3}
Every twin-dh digraph has directed clique-width at most $3$.
\end{theorem}

\begin{proof}
  We show a construction for a directed clique-width $3$-expression for every $G\in DDH$ and then argue, why this is best possible.
  The method we use is closely related to the undirected case: Distance-hereditary graphs have clique-width at most $3$, see \cite{GR00}.
  Let $G\in \DDH$ and $S = (s_1, \dots, s_n)$ be a directed pruning sequence creating $G$.
  We give an algorithm to construct a $3$-expression traversing $S$ starting with the last element of the sequence.
  The idea for computing this expression is to use three labels $1,2$ and $3$ as follows: After every step of the algorithm, expressions which are already constructed consist of vertices labeled with $1$ and $3$, where $1$ means that the vertex has not been finally treated and possibly has edges to other vertices not inserted yet, whereas for vertices labeled by $3$ all incident edges have already been considered. The label $2$ is only used as a working label.
  For initialization, let $X_v= \bullet_1$ for every vertex $v \in V$. As this is a $1$-expression, it is also a $3$-expression.
  In the following let now $s_i = (v_i, op_i, a_i)$, for simplification denoted by $(v, op, a)$, be the currently treated element of $S$ and $X_v$ and $X_a$ be the $3$-expressions which exists by induction for $v$ and $a$. Then, we get a $3$-expression by the following rules depending on the operation $op$.
  \begin{itemize}
    \item[(1)] $op =~\ppv : ~~ X_a := \rho_{2\to 3} ( \alpha_{2,1}( \rho_{1 \to 2} (X_v) \oplus X_a  ))$
    \item[(2)] $op =~\pmv : ~~ X_a :=  \rho_{2\to 3}(\alpha_{1,2}(\rho_{1 \to 2} (X_v) \oplus X_a))$
    \item[(3)] $op =~~\ft : ~~~  X_a := X_v \oplus X_a$
    \item[(4)] $op =~\tit : ~~  X_a := \rho_{2\to 1}(\alpha_{2,1}(\rho_{1 \to 2} (X_v) \oplus X_a))$
    \item[(5)] $op =~\tot : ~~  X_a := \rho_{2\to 1}(\alpha_{1,2}(\rho_{1 \to 2} (X_v) \oplus X_a))$
    \item[(6)] $op =~\tdt :~~  X_a := \rho_{2\to 1}(\alpha_{1,2}(\alpha_{2,1}(\rho_{1 \to 2} (X_v) \oplus X_a)))$
  \end{itemize}
To prove correctness we first need some definitions. For $w$ a vertex of $V(G)$, let $G(w)_i$ be the graph consisting of $w$ and every vertex that is generated by operations on $w$ after step $i$, which means that $G(w)_i$ is created by the pruning sequence $S(w)_i$ which contains elements $s_k = (v_k, op_k, v_{a_k})$ with $k\geq i$ and $v_{a_k}$ has been generated by a series of operations by $w$.

For $i = n$, this means that $G(w)_i = (\{w\}, \emptyset )$.
Note that for element $v_0$, which is the first anchor in the pruning sequence, i.e. $s_1 = (v_1, op_1, v_0)$, it holds that $S(v_0)_1 = S$ and $G(v_0)_1=G$.

We then show by induction that at any step $i$ with $n \geq i \geq 0$ of the algorithm, it holds that $X_w$ is a $3$-expression of $G(w)_i$ for all vertices $w \in V(G)$.
We further assume that every $X_w$ contains only vertices labeled by $1$ and $3$, where the vertices labeled by $1$ are exactly those, which are created by a series of twin operations on $w$ (including $w$ itself). To simplify, we call such a vertex a \emph{far twin} in the following.

After the initialization, it is easy to see for $i=n$ that for all $w \in V(G)$, $X_w= \bullet_1$ is a $3$-expression of $G(w)_i = (\{w\}, \emptyset )$. Obviously, the only vertex in  $G(w)_i$ is $w$ which is labeled by $1$.

Consider now step $i$. By induction we know that $X_a$ is a $3$-expression of $G(a)_{i+1}$ where far twins of $a\in V(G)$ are labeled by $1$ and all other vertices are labeled by $3$. Further, $X_v$ is a $3$-expression of $G(v)_{i+1}$ where all far twins of $v\in V(G)$ are labeled by $1$ and all other vertices are labeled by $3$. We now show that after step $i$ it holds that $X_a$ is a $3$-expression of $G(a)_i$. Therefore, we consider element $s_i=:(v, op, a)$ in step $i$.

\begin{itemize}
\item[(1)] As $v$ is a pendant plus vertex of $a$, there exist edges from every far twin of $v$ to every far twin of $a$.
By $\rho_{1 \to 2} (X_v)$ we relabel every vertex in $X_v$ which is labeled by $1$, i.e. every far twin of $v$ with $2$.
We join this expression with $X_a$ and add edges from all labels $2$ to $1$, which inserts all edges created by the pendant plus relation of $v$ to $a$.
As $v$ is a pendant plus vertex of $a$, all far twins of $a$ can not be far twins of $v$ and thus, we relabel these vertices from $2$ to $3$.
Now, the newly created $X_a$ is a $3$-expression of $G(a)_i$ consisting only of labels $1$ and $3$, where the vertices labeled by $1$ are exactly the far twins of $a$.
\item[(2)] Analogously to (1).
\item[(3)] As $v$ is a false twin of $a$, no new edges are inserted by this operation and further all far twins of $v$ are also far twins of $a$.
Therefore, we only need to join expressions $X_v$ and $X_a$ to create an expression for $G(a)_i$ where every far twin of $a$ is labeled by $1$ and every other vertex is labeled by $3$.
\item[(4)] As $v$ is a true in-twin of $a$, like in (1), there exist edges from every far twin of $v$ to every far twin of $a$.
We therefore use the same method to join the expressions $X_v$ and $X_a$ and create edges between them.
But unlike in (1), every far twin of $v$ is a far twin of $a$. Therefore, we relabel the far twins of $v$ from $2$ to $1$, to obtain a $3$-expression for $G(a)_i$ in which exactly all far twins of $a$ are labeled by $1$.
\item[(5)] Analogously to (4).
\item[(6)] Is very similar to (4) and (5), with the only difference that we need edges from every far twin of $v$ to every far twin of $a$ and the other way round.
Therefore, we insert edges from labels $2$ to $1$ and from labels $1$ to $2$, before relabeling, to obtain a $3$-expression for $G(a)_i$ in which exactly all far twins of $a$ are labeled by $1$.
\end{itemize}

Further, the directed clique-width of a twin-dh digraph has to be at least $3$, as can be seen by the following counterexample: The $\overrightarrow{P_3}$, which means a directed path of $3$ vertices, is twin-distance-hereditary, but it is not expressible by a $2$-expression.
\end{proof}

By Courcelles Theorem on clique-width, the bounded directed clique-width of these digraphs leads to the computability of every problem, which is describable in monadic second order logic.

\begin{corollary}
Let $G$ be a twin-dh digraph. Then every graph problem, which is describable in $\mso_1$ logic, is computable in polynomial time on $G$.
\end{corollary}

%%%%%%%%%%%%%%%%%%%%%%%%%%%%%%%%%%%%%%%%%%%%%%%%%%%%%%%%%%%%%%%%%%%%%%%%%%%%%%%%%%%%%%%%%%%%%%%%%%%
\section{Further Problems on directed twin-dh graphs}
%%%%%%%%%%%%%%%%%%%%%%%%%%%%%%%%%%%%%%%%%%%%%%%%%%%%%%%%%%%%%%%%%%%%%%%%%%%%%%%%%%%%%%%%%%%%%%%%%%%

From the result in \cite{GKR21} we can follow, that the dichromatic number problem can be solved in polynomial time on twin-dh digraphs.

By \cite{GWY16} we can solve the problems Directed Hamiltonian Path, Directed Hamiltonian Cycle, Directed Cut, and Regular Subdigraph using an XP-algorithm w.r.t. the parameter directed NLC-width in polynomial time.
Directed NLC-width is a digraph parameter which is closely related to directed clique-width, since we can transform every directed clique-width $k$-expression into an equivalent NLC-width $k$-expression, see \cite{GWY16}.
Thus, directed twin-dh graphs have bounded NLC-width and we can solve the above mentioned problems in polynomial time on this class.

%%%%%%%%%%%%%%%%%%%%%%%%%%%%%%%%%%%%%%%%%%%%%%%%%%%%%%%%%%%%%%%%%%%%%%%%%%%%%%%%%%%%%%%%%%%%%%%%%%%
\section{Conclusion and Outlook}
%%%%%%%%%%%%%%%%%%%%%%%%%%%%%%%%%%%%%%%%%%%%%%%%%%%%%%%%%%%%%%%%%%%%%%%%%%%%%%%%%%%%%%%%%%%%%%%%%%%

In this paper, we introduce directed twin-distance-hereditary graphs, which are developed by a generalization of the recursive definition for undirected distance-hereditary graphs.
The class of twin-dh digraphs is a superclass of directed co-graphs, defined in \cite{CP06} and when excluding the bioriented true twin operation, it is a subclass of distance-hereditary digraphs, defined in \cite{LR10}.
We characterize this class by forbidden induced subdigraphs.
Further, we show that the class is a subclass of extended directed co-graphs which allows us to deduce interesting results.
However, due to the unbounded directed clique-width of extended directed co-graphs, twin-dh digraphs exhibit properties which allow supplemental results such that an investigation of this class is advisable.
We show interesting properties of the class such as that every strong components is a directed co-graph.
This is helpful in the computation of several problems.
For future work it might be interesting to investigate if problems which are solvable on undirected distance-hereditary graphs, as e.g. the (directed) Steiner tree problem can also be solved on twin-dh digraphs.

Moreover, we show that several directed width parameters, namely directed path-width, directed tree-width, DAG-width and cycle rank can be computed in linear time on directed twin-distance-hereditary graphs. From the associated proof (as well as from the fact that this twin-dh digraphs are a subclass of extended directed co-graphs) further the equality of all these parameters follows.

Further, we can conclude by our results and \cite{GKR21a} that for twin-dh digraphs, as for directed co-graphs, Kelly-width can be bounded by DAG-width.
Due to \cite[Conjecture 30]{HK08}, \cite{AKKRS15}, and
\cite[Section 9.2.5]{BG18} this remains open for general digraphs and is related to one of the
biggest open problems in graph searching, namely whether the monotonicity
costs for Kelly- and DAG-width games are bounded.  In our previous paper we could show, that
the equivalence of those two parameters is given on directed co-graphs. In this paper we can
extend this result to their superclass of twin-dh digraphs.

Like in the undirected case, every twin-dh digraph has directed clique-width at most $3$, though not every digraph of directed clique-width $3$ is a twin-dh digraph.
From that we can conclude  several interesting results, since there are many NP-hard problems which are solvable on digraphs of bounded directed clique-width.

It would be interesting for future work to consider other superclasses of twin-dh digraphs and whether it is still possible to find efficient algorithms to compute several graph parameters on these classes and at which point it becomes NP-hard.

%%%%%%%%%%%%%%%%%%%%%%%%%%%%%%%%%%%%%%%%%%%%%%%%%%%%%%%%%%%%%%%%%%%%%%%%%%%%%%%%%%%%%%%%%%%%%%%%%%%
\paragraph{\textbf{Acknowledgements}}
%%%%%%%%%%%%%%%%%%%%%%%%%%%%%%%%%%%%%%%%%%%%%%%%%%%%%%%%%%%%%%%%%%%%%%%%%%%%%%%%%%%%%%%%%%%%%%%%%%%
We thank Frank Gurski and Van Bang Le for the very useful discussions. \\
The work of both authors was supported by the German Research Association (DFG)
grant GU 970/7-1. \\

% Bibliographie
\bibliographystyle{plain}

%%%%%%%%%%%%%%%%%%%%%%%%%%%%%%%%%%%%
\end{document}